\documentclass[10pt]{article}
\pdfoutput=1

\usepackage{algorithm,algpseudocode}
\usepackage{amsmath,amsthm}
\usepackage{enumitem}
\usepackage{geometry}
\usepackage{mathtools}
\usepackage{subcaption}

\mathtoolsset{showonlyrefs=true}

\DeclareMathOperator{\polylog}{polylog}

\title{Deterministic and Work-Efficient\\Parallel Batch-Dynamic Trees in Low Span}

\author{ \normalsize Daniel Anderson \\ \normalsize Carnegie Mellon University \\ \normalsize dlanders@cs.cmu.edu \and \normalsize Guy E. Blelloch \\ \normalsize Carnegie Mellon University \\ \normalsize guyb@cs.cmu.edu}

\date{}

\newcommand{\myparagraph}[1]{\medskip{\noindent\bfseries\itshape{#1}.~}}

\newtheorem{theorem}{Theorem}
\newtheorem{definition}{Definition}

\newtheorem{corollary}{Corollary}
\newtheorem{lemma}{Lemma}

\begin{document}

\maketitle

\begin{abstract}
    Dynamic trees are a well-studied and fundamental building block of dynamic graph algorithms dating back to the seminal work of
    Sleator and Tarjan [\textit{STOC'81}, (1981), pp.~114-122]. The problem is to maintain a tree subject to online edge insertions and
    deletions while answering queries about the tree, such as the heaviest weight on a path, etc.
    In the parallel batch-dynamic setting, the goal is to process batches of edge updates work efficiently in low ($\polylog n$) span.
    Two work-efficient algorithms are known, batch-parallel Euler Tour Trees by Tseng et al.\ [\textit{ALENEX'19}, (2019), pp.~92--106]
    and parallel Rake-Compress (RC) Trees by Acar et al.\ [\textit{ESA'20}, (2020), pp.~2:1--2:23]. Both however are randomized and
    work efficient in expectation. Several downstream results that use these data structures (and indeed to the best of our knowledge, all known
    work-efficient parallel batch-dynamic graph algorithms) are therefore also randomized.
    
    In this work, we give the first deterministic
    work-efficient solution to the problem. Our algorithm maintains a dynamic parallel tree contraction subject to batches of $k$ edge updates deterministically in
    worst-case $O(k \log(1 + n/k))$ work and $O(\log n \log^{(c)} k)$ span for any constant $c$. This allows us to implement parallel batch-dynamic
    RC-Trees with worst-case $O(k \log(1 + n/k))$ work updates and queries deterministically. Our techniques that we use to obtain the given span bound
    can also be applied to the state-of-the-art randomized variant of the algorithm to improve its span from $O(\log n \log^* n)$ to $O(\log n)$.
\end{abstract}

\clearpage

\section{Introduction}

The dynamic trees problem dates back to the seminal work of Sleator and Tarjan~\cite{sleator1983data} on \emph{Link-Cut trees}. The problem is to maintain a forest of trees subject to the insertion and deletion of edges, while answering queries about the forest. Examples of supported queries include connectivity (is there a path from $u$ to $v$, i.e., are $u$ and $v$ in the same tree), the weight of all vertices in a specific subtree, and the weight of the heaviest (or lightest) edge on a path. The latter is a key ingredient in the design of efficient algorithms for the maximum flow problem~\cite{sleator1983data,goldberg1988new,ahuja1989improved,goldberg1991use,tarjan1997dynamic}, which was the motivation for their invention. Dynamic trees are also ingredients in algorithms for dynamic graph connectivity~\cite{frederickson1985data,henzinger1995randomized,holm2001poly,kapron2013dynamic,acar2019parallel} dynamic minimum spanning trees~\cite{frederickson1985data,henzinger1995randomized,henzinger2001maintaining,holm2001poly,anderson2020work}, and minimum cuts~\cite{karger2000minimum,geissmann2018parallel,gawrychowski2019minimum,anderson2021cuts}.

There are a number of efficient ($O(\log n)$ time per operation) dynamic tree algorithms, including Sleator and Tarjan's \emph{Link-Cut Tree}~\cite{sleator1983data,sleator1985self}, Henzinger and King's \emph{Euler-Tour Trees}~\cite{henzinger1995randomized}, Frederickson's \emph{Topology Trees}~\cite{frederickson1985data,frederickson1997ambivalent,frederickson1997data}, Holm and de Lichtenberg's \emph{Top Trees}~\cite{holm1998top,tarjan2005self,alstrup2005maintaining}, and Acar et al.'s Rake-Compress Trees~\cite{acar2004dynamizing,acar2005experimental,acar2020batch}. Most of these algorithms are sequential and handle single edge updates at a time.

Two exceptions are Tseng et al.'s \emph{Batch-Parallel Euler-Tour Trees}~\cite{tseng2019batch} and Acar et al.'s \emph{Parallel Batch-Dynamic RC-Trees}~\cite{acar2020batch}. These algorithms implement \emph{batch-dynamic} updates, which take a set of $k$ edges to insert or delete with the goal of doing so in parallel. Both of these algorithms achieve work-efficient $O(k \log(1+n/k))$ work, which matches the sequential algorithms ($O(\log n)$ work) for low values $k$, and is optimal ($O(n)$ work) for large values of $k$. Both of these algorithms, however, are randomized.

Other graph problems have also been studied in the parallel batch-dynamic model, such as connectivity~\cite{acar2019parallel}, minimum spanning trees~\cite{ferragina1994batch,ferragina1996three,pawagi1993optimal,shen1993parallel,anderson2020work,tseng2022parallel}, and approximate $k$-core decomposition~\cite{liu2022parallel}. However, to the best of our knowledge, all efficient parallel batch-dynamic graph algorithms are randomized. Indeed, avoiding randomization seems difficult even for some classic static problems. Finding a spanning forest, for instance, has a simple $O(m)$-time sequential algorithm, and an $O(m)$ work, $O(\log n)$ span randomized parallel algorithm has been known for twenty years~\cite{pettie2002randomized}, but no deterministic equivalent has been discovered. The best deterministic algorithm requires an additional $\alpha(n,m)$ factor of work~\cite{cole1991approximate}.

\myparagraph{Our results}
We design a work-efficient algorithm for the batch-dynamic trees problem that is deterministic, $\polylog n$ span, and doesn't use concurrent writes.

\begin{theorem}
There is a deterministic parallel batch-dynamic algorithm that maintains a balanced Rake-Compress Tree (RC-Tree) of a bounded-degree forest subject to batches of $k$ edge updates (insertions, deletions, or both) in $O(k \log(1+n/k))$ work and $O(\log(n) \log^{(c)} k)$ span\footnote{The notation $\log^{(c)} n$ refers to the \emph{repeated/nested logarithm} function, i.e., $\log^{(1)} n = \log n$ and $\log^{(c)} n = \log( \log^{(c-1)} n )$. E.g., $\log^{(2)} n = \log \log n$.} for any constant $c$.
\end{theorem}

\noindent The resulting RC-Tree is amenable to all existing query algorithms for parallel RC-Trees and hence it can solve batch connectivity queries, batch subtree queries, batch path queries, batch LCA queries and non-batched diameters, centers and medians~\cite{acar2005experimental}.

\begin{corollary}
    There are deterministic parallel batch-dynamic algorithms for
    \begin{enumerate}
        \item Dynamic forest connectivity
        \item Subtree sums of an associative commutative operation over vertices or edges
        \item Path sums of an invertible associative commutative operation over vertices or edges
        \item Path minimum- or maximum-weight edge
        \item Lowest common ancestors
    \end{enumerate}
    running in $O(k \log(1+n/k))$ work and $O(\log n)$ span
    for batches of $k$ queries.
\end{corollary}

\noindent We start by developing a simpler algorithm that is work efficient but runs in $O(\log n \log \log k)$ span, then, in Section~\ref{sec:improved-span}, we discuss how to optimize the span. As a byproduct of our span optimization,
we also optimize the randomized variant of the problem and obtain the following improved result, which improves over the $O(\log n \log^* n)$ span of Acar et al.~\cite{acar2020batch}.

\begin{theorem}
There is a randomized parallel batch-dynamic algorithm that maintains a balanced Rake-Compress Tree (RC-Tree) of a bounded-degree forest subject to updates of $k$ edges in $O(k \log(1+n/k))$ work and $O(\log(n))$ span.
\end{theorem}

\noindent Randomized parallel batch-dynamic trees have already been used as key ingredients in several other parallel batch-dynamic algorithms. As an additional result, existing algorithms for fully dynamic connectivity~\cite{acar2019parallel} and incremental minimum spanning trees~\cite{anderson2020work} can be derandomized by using our deterministic RC-Tree, assuming that the underlying graph is ternarized (transformed to constant degree~\cite{frederickson1985data}). Details are given in the appendix.

\begin{theorem}
There is a deterministic parallel batch-dynamic algorithm which, given batches of $k$ edge insertions, deletions, and connectivity queries on a bounded-degree graph
processes insertions and deletions in $O(k \log n \log(1 + n/\Delta) \alpha(n,m))$ amortized work and $O(\polylog n)$ span, and answers
queries in \mbox{$O(k \log(1 + n/k))$} work and $O(\polylog n)$ span, where $\Delta$ is the average batch size of all deletion operations.
\end{theorem}

\begin{theorem}
There is a deterministic parallel batch-incremental algorithm which maintains the minimum spanning forest of a weighted undirected bounded-degree graph
given batches of $k$ edge insertions in $O(k \log (1 + n/k) + k \log \log n)$ work and $O(\polylog n)$ span.
\end{theorem}

\myparagraph{Overview} Previous implementations of RC-Trees are all based on applying self-adjusting computation to randomized parallel tree contraction~\cite{acar2004dynamizing,acar2005experimental},
i.e., a static tree contraction algorithm is implemented in a framework that automatically tracks changes to the input values, and selectively recomputes procedures that depend on changed data. This results in the RC-Tree that would have been obtained if running the static algorithm from scratch on the updated data. In
our case, we instead describe a direct update algorithm that is not based on self-adjusting computation. We describe a variant of tree contraction that deterministically
contracts a maximal independent set (MIS) of degree one and two vertices each round. When an update is made to the forest, we identify the set of \emph{affected vertices}
(this part is similar to change propagation) and then greedily update the tree contraction by computing an MIS of affect vertices and updating
the contraction accordingly.  The key insight is in correctly establishing the criteria for vertices being affected such that the update is correct, and
bounding the number of such vertices so that it is efficient.

\begin{figure*}[th]
\centering
\begin{subfigure}{0.49\textwidth}
  \centering
  \includegraphics[width=0.99\textwidth]{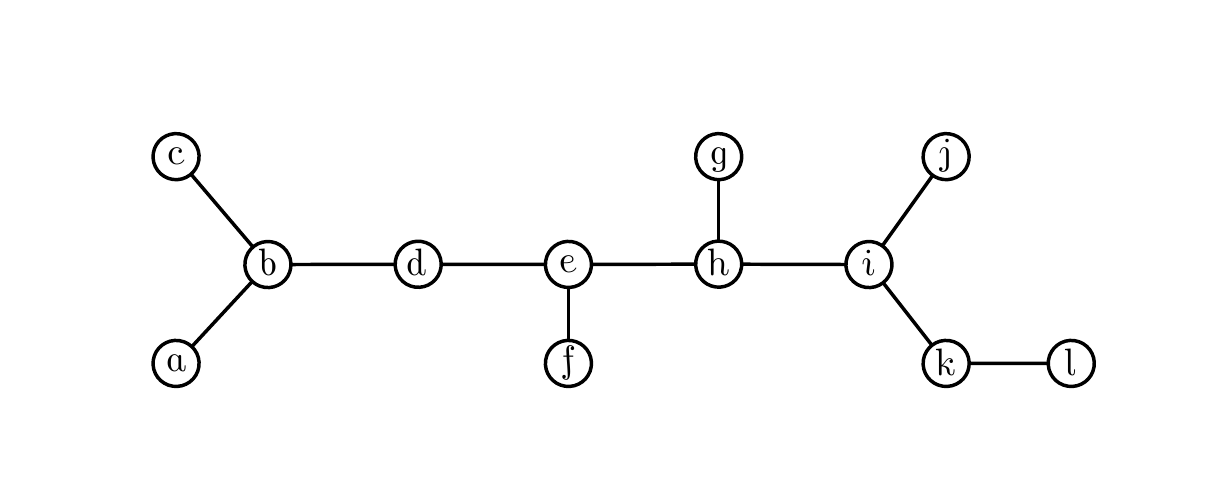}
  \caption{An unrooted tree}
\end{subfigure}
\begin{subfigure}{0.5\textwidth}
  \centering
  \includegraphics[width=0.99\textwidth]{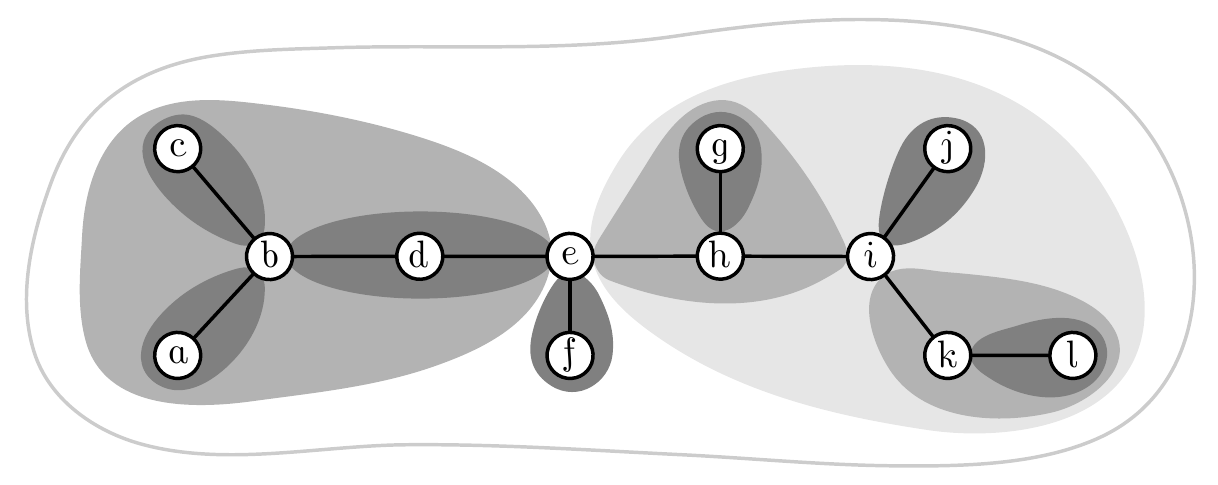}
  \caption{A recursive clustering of the tree produced by tree
    contraction. Clusters produced in earlier rounds are depicted in
    a darker color.}
\end{subfigure}

\bigskip

\begin{subfigure}{0.8\textwidth}
  \centering
  \includegraphics[width=0.85\textwidth]{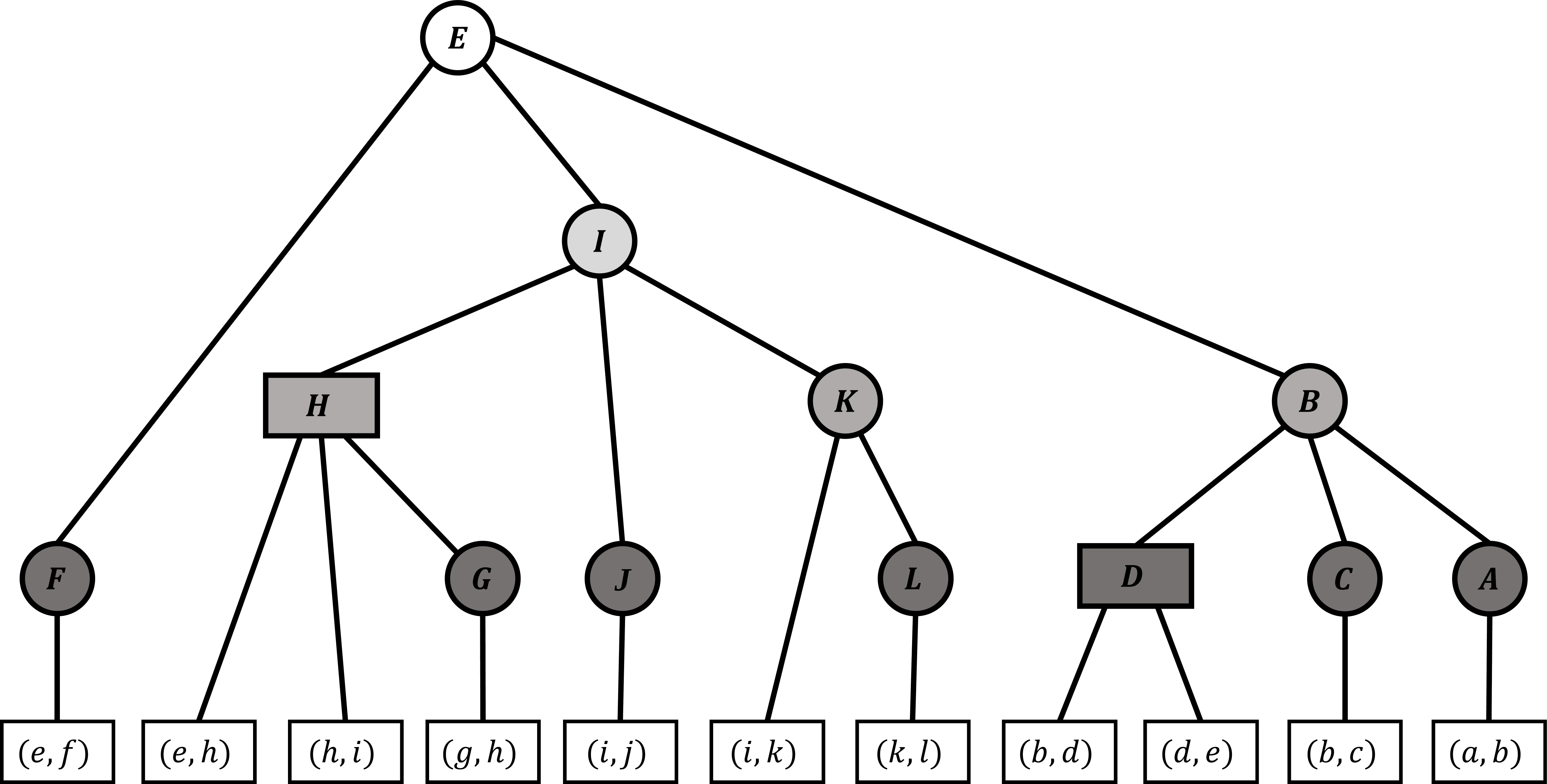}
  \caption{The corresponding RC-Tree. Clusters produced from rakes (and the root cluster) are shown as filled circles, and clusters produced from compress as rectangles. The base clusters (edges of the original tree) are labeled in lowercase, and the composite clusters are labeled with the uppercase of their representative vertex. The shade of a cluster corresponds to its height in the clustering. Lower heights (i.e., contracted earlier) are darker.}
\end{subfigure}
\caption{A tree, a clustering, and the corresponding RC-Tree \cite{acar2020batch}.}\label{fig:rc-tree}
\end{figure*}

\section{Preliminaries}

\subsection{Model of computation}

\myparagraph{Parallelism}
We analyze algorithms in the shared-memory \emph{work-span} model with \emph{fork-join} parallelism. A procedure can \emph{fork} a set of procedures to run in parallel and wait for forked procedures to complete with a \emph{join}. The work of the algorithm is the total number of instructions performed, and the span (also called depth or parallel time) is the length of the longest chain of sequentially dependent instructions.
An alternative model is the classic PRAM model, which consists of $p$ processors that execute in lock step with access to a shared memory. An algorithm that runs in $t$ steps on a $p$-processor PRAM performs $O(pt)$ work in $O(t)$ span. Our algorithms are naturally described as fork-join programs, but we will use many classic PRAM algorithms as subroutines.

\myparagraph{Concurrency}
A parallel computation may be endowed with the ability to read/write the shared memory concurrently. In the classic PRAM model, algorithms are classified into one of three categories: \textsc{EREW} (\emph{Exclusive-Read Exclusive-Write}, i.e., no concurrency allowed at all), \textsc{CREW} (\emph{Concurrent-Read Exclusive-Write}, i.e., concurrent reads are allowed but concurrent writes are prohibited), and \textsc{CRCW} (Concurrent-Read Concurrent-Write, i.e., concurrent writes are permitted). The \textsc{CRCW} model may be further subdivided by the behaviour of concurrent writes. The \emph{Common} \textsc{CRCW PRAM} permits concurrent writes but requires that all concurrent writes to the same location write the same value, else the computation is invalid. The \emph{Arbitrary} \textsc{CRCW PRAM} permits concurrent writes of different values and specifies that an arbitrary processor's write succeeds, but the algorithm may make no assumption about which processor succeeds. Lastly, the \emph{Priority} \textsc{CRCW PRAM} assigns fixed unique
priorities to each processor, and specifies that the highest priority processor's write succeeds.

\subsection{Colorings and maximal independent sets}

Parallel algorithms for graph coloring and maximal independent sets are well studied~\cite{luby1985simple,jung1988parallel,cole1986approximate,goldberg1987parallel,goldberg1987parallel2}. Our algorithm uses a subroutine that finds a maximal independent set of a collection of chains, i.e., a set of vertices of degree one or two. Goldberg and Plotkin~\cite{goldberg1987parallel2} give an algorithm that finds an $O(\log^{(c)}(n))$-coloring of a constant-degree graph in $O(c\cdot n)$ work and $O(c)$ span in the EREW model. This gives a constant coloring in $O(n \log^* n)$ work and $O(\log^* n)$ span.

Given a $c$-coloring, one can easily obtain a maximal independent set as follows. Sequentially, for each color, look at each vertex of in parallel. If a vertex is the current color and is not adjacent to a vertex already selected for the independent set, then select it. This takes $O(c \cdot n)$ work and $O(c)$ span, hence a constant coloring yields an $O(n)$ work and constant span algorithm.

For any choice of coloring, the above algorithms do not yield a work-efficient algorithm for maximal independent set since it is not work efficient to find a constant coloring, and not work efficient to convert a non-constant coloring into an independent set. To make it work efficient, we can borrow a trick from Cole and Vishkin. We first produce a $\log^{(c)} n$ coloring in $O(1)$ time, then bucket sort the vertices by color, allowing us to perform the coloring-to-independent-set conversion work efficiently. Cole and Vishkin~\cite{cole1986deterministic} show that the bucket sorting can be done efficiently in the EREW model. This results in an $O(n)$ work, $O(\log^{(c)} n)$ span algorithm for maximal independent set in a constant-degree graph.

\begin{lemma}[\cite{goldberg1987parallel2,cole1986deterministic}]\label{lem:mis-in-constant-degree-graph}
    There exists a deterministic algorithm that finds an MIS in a constant-degree graph in $O(n)$ work and $O(\log^{(c)} n)$ span for any
    constant $c$ in the EREW model.
\end{lemma}

\subsection{Filtering and compaction}

A \emph{filter} is a fundamental primitive used in many parallel algorithm. Given an input array $A$ of length $n$ and a predicate $p$,
a filter returns a new array consisting of the elements of $A$ that satisfy the predicate. Filtering can easily be
implemented in linear work and $O(\log n)$ span. This span however is a bottleneck for many applications, so \emph{approximate compaction}
was developed as an alternative. Approximate compaction takes the same input as a filter, such that if
there are $m$ elements that satisfy the predicate, the output is an array of size $O(m)$ containing the elements that satisfy
the predicate, and some blank elements. Goldberg and Zwick~\cite{goldberg1995optimal} give a deterministic approximate compaction algorithm that is work efficient
and runs in $O(\log \log n)$ span in the Common CRCW model.

\begin{lemma}[\cite{goldberg1995optimal}]\label{lem:approximate-compaction}
    There exists a deterministic algorithm that performs approximate compaction on an array of $n$ elements in
    $O(n)$ work and $O(\log \log n)$ span in the Common CRCW model.
\end{lemma}

\subsection{Parallel tree contraction}

Tree contraction is a procedure for computing functions over trees in parallel in low span~\cite{miller1985parallel}. It involves repeatedly applying \emph{rake} and \emph{compress} operations to the tree while aggregating data specific to the problem. The rake operation removes a leaf from the tree and aggregates its data with its parent. The compress operation replaces a vertex of degree two and its two adjacent edges with a single edge joining its neighbors, aggregating any data associated with the vertex and its two adjacent edges.

Rake and compress operations can be applied in parallel as long as they are applied to an independent set of vertices. Miller and Reif~\cite{miller1985parallel} describe a linear work and $O(\log n)$ span randomized algorithm that performs a set of rounds, each round raking every leaf and an independent set of degree two vertices by flipping coins. They show that it takes $O(\log n)$ rounds to contract any tree to a singleton with high probability. They also describe a deterministic algorithm but it is not work efficient. Later, Gazit, Miller, and Teng~\cite{gazit1988optimal} improve it to obtain a work-efficient deterministic algorithm with $O(\log n)$ span.

These algorithms are defined for constant-degree trees, so non-constant-degree trees are handled by converting them into bounded-degree equivalents, e.g., by ternerization~\cite{frederickson1985data}.

\subsection{Rake-Compress Trees (RC-Trees)}

The process of tree contraction can alternatively be viewed as a clustering of the underlying tree. A cluster is a connected subset of edges and vertices of the tree. The base clusters are the individual edges and vertices of the tree.  A cluster may contain an edge without containing the endpoints of that edge. Such a vertex is called a \emph{boundary vertex} of the cluster. Every cluster has at most two boundary vertices.

To form a recursive clustering from a tree contraction, we begin with the base clusters and the uncontracted tree. On each round, for each vertex $v$ that contracts via rake or compress (which remember, form an independent set), we identify the set of clusters that share $v$ as a common boundary vertex. These clusters are merged into a single cluster consisting of the union of their contents. We call $v$ the \emph{representative} vertex of the resulting cluster. The boundary vertices of the resulting cluster will be the union of the boundary vertices of the constituents, minus $v$. Clusters originating from rake operations have one boundary vertex and are called
\emph{unary clusters}, and clusters originating from compresses have two boundary vertices and are called \emph{binary clusters}. Since each vertex contracts exactly once, there is a one-to-one mapping between representative vertices of the original tree and non-base clusters.

An RC-Tree encodes this recursive clustering. The leaves of the RC-Tree are the base edge clusters of the tree, i.e., the original edges (note that the base cluster for vertex $v$ is always a direct child of the cluster for which $v$ is the representative, so omitting it from the RC-Tree loses no information). Internal nodes of the RC-Tree are clusters formed by tree contraction, such that the children of a node are the clusters that merged to form it. The root of the RC-Tree is a cluster representing the entire tree.

Queries are facilitated by storing augmented data on each cluster, which is aggregated from the child clusters at the time of its creation. Since tree contraction removes a constant fraction of the vertices at each round, the resulting RC-Tree is balanced, regardless of how balanced or imbalanced the original tree was. This allows queries to run in $O(\log n)$ time for single queries, or \emph{batch queries} to run in $O(k \log(1+n/k))$ work and $O(\log n)$ span.

An example of a tree, a recursive clustering induced by a tree contraction, and the corresponding RC-Tree are depicted in Figure~\ref{fig:rc-tree}.

\section{Algorithms}

\subsection{The data structure}

Our algorithm maintains a \emph{contraction data structure} which serves to record the
process of tree contraction beginning from the input forest $F$ as it contracts to a forest of singletons. The contraction
data structure also contains the RC-Tree as a byproduct.  For each round
in which a vertex is live, the contraction data structure stores an adjacency list for that vertex. Since the forests have bounded degree $t$, each
adjacency list is exactly $t$ slots large.  Each entry in a vertex $v$'s adjacency list is one of
four possible kinds of value:
\begin{enumerate}[leftmargin=*]
    \item Empty, representing no edge
    \item A pointer to an edge of $F$ adjacent to $v$
    \item A pointer to a binary cluster for which $v$ is one of the boundary vertices. This represents an edge
    between $v$ and the other boundary vertex in the contracted tree.
    \item A pointer to a unary cluster for which $v$ is the boundary vertex. This does not represent any edge
    in the contracted tree, but is used to propagate augmented data from the child to the parent.
\end{enumerate}

\noindent At round $0$ (before any contraction), the adjacency list simply stores pointers to the edges adjacent to each vertex.
At later rounds, in a partially contracted tree, some of the edges are not original edges of $F$, but are the result
of a compress operation and represent a binary cluster of $F$ (Case 3), which may contain augmented data (e.g., the sum of the weights in the
cluster, the maximum weight edge on the cluster path, etc, depending on the application). Additionally, vertices that rake
accumulate augmented data inside their resulting unary cluster that needs to be aggregated when their parent cluster is created (Case 4).
Acar et al.~\cite{acar2005experimental} describe many examples of the kinds of augmented data that are used in common applications.

Clusters contain pointers to their child clusters alongside any augmented data. Each composite cluster
corresponds uniquely to the vertex that contracted to form it, so counting them plus the base edge clusters, the RC-Tree contains exactly
$n + m$ clusters. If the user wishes to store augmented data on the vertices, this can be stored on the unique composite cluster for which
that vertex is the representative.

\subsection{The static algorithm}

We build an RC-Tree deterministically using a variant of Miller and Reif's tree contraction algorithm.  Instead of contracting all degree-one vertices (leaves) and an independent set of degree two vertices, we instead contract any maximal independent set of degree one and two vertices, i.e., leaves are not all required to contract.
The reason for this will become clear during the analysis of the update algorithm, but essentially, not forcing leaves to contract reduces the number of vertices that
need to be reconsidered during an update, since a vertex that was previously not a leaf becoming a leaf would force it to contract, which would force its neighbor
to not contract, which may force its other neighbor to contract to maintain maximality. Such a chain reaction is undesirable, and our variant avoids it.

\begin{definition}
    We say that a tree contraction is maximal at some round if the set of vertices that contract form
    a maximal independent set of degree one and two vertices.  A tree contraction is maximal if it is
    maximal at every round.
\end{definition}

\noindent Let us denote by $F_0$, the initial forest, then by $F_i$ for $i \geq 1$, the forest obtained by applying one round
of maximal tree contraction to $F_{i-1}$. To obtain a maximal tree contraction of $F_i$, consider in parallel
every vertex of degree one and two. These are the \emph{eligible} vertices. We find a maximal independent
set of eligible vertices by finding an $O(\log \log n)$ coloring of the vertices using the algorithm of Goldberg
and Plotkin~\cite{goldberg1987parallel2}, then bucket sorting the vertices by colors and selecting a maximal
independent set similar to the algorithm of Cole and Vishkin~\cite{cole1986deterministic}. This takes $O(|F_i|)$
work and $O(\log \log n)$ span. To write down the vertices of $F_{i+1}$, we can apply approximate compaction
to the vertex set of $F_i$, filtering those which were selected to contract. This also takes $O(|F_i|)$ work
and $O(\log \log n)$ span~\cite{goldberg1995optimal}.

As is standard for parallel tree contraction implementations, each vertex writes into its
neighbors' adjacency list for the next round. For vertices that do not contract, they can
simply copy their corresponding entry in their neighbors' adjacency list to the next round.
For vertices that do contract, they write the corresponding cluster pointers into their neighbors'
adjacency lists. For example, if a vertex $v$ with neighbors $u$ and $w$ compresses, it writes
a pointer to the binary cluster formed by $v$ (whose boundary vertices are $u$ and $w$) into the
adjacency list slots of $u$ and $w$ that currently store the edge to $v$. The tree
has constant degree, so identifying the slot takes constant time and suffers
no issues of concurrency.

To build the RC-Tree clusters, it suffices to observe that when a vertex $v$ contracts, the contents of
its adjacency list are precisely the child clusters of the resulting cluster. Hence in constant
time we can build the cluster by aggregating the augmented values of the children and creating
a corresponding cluster.  If at a round, a vertex is isolated (has no neighbors), it creates a root cluster.

\subsection{Dynamic updates}

\begin{algorithm*}[th]
  \caption{Batch update}
  \label{alg:update}
  \footnotesize
  \begin{algorithmic}[1]
    \Procedure{BatchUpdateForest}{$E^+, E^-$}
    \State Update $F_0$, adding edges in $E^+$ and removing edges in $E^-$
    \State Determine affected vertices from the endpoints of $E^+ \cup E^-$
    \For{each level $i$ from $1$ to $\log_{6/5} n$}
        \State\label{line:new_affected} Determine the new affected vertices from the previous set of affected vertices
        \State\label{line:new_mis} Find a maximal set of eligible affected vertices that do not have an unaffected neighbor that contracts in $F_i$
        \State\label{line:update_tree} Obtain $F_{i}'$ from $F_i$ by uncontracting all affected vertices, then contracting the vertices in the new MIS
    \EndFor
    \EndProcedure
  \end{algorithmic}
\end{algorithm*}

\noindent A dynamic update consists of a set of $k$ edges to be added or deleted (a combination of both is valid).
The update begins by modifying the adjacency lists of the $2k$ endpoints of the modified edges. Call this
resulting forest $F_0'$, to denote the forest after the update. The goal of the update algorithm is now
to produce $F_i'$, an updated tree contraction for each level $i$, using $F_i, F_{i-1}$ and $F_{i-1}'$.

\begin{figure}
    \centering
    \includegraphics[width=0.4\columnwidth]{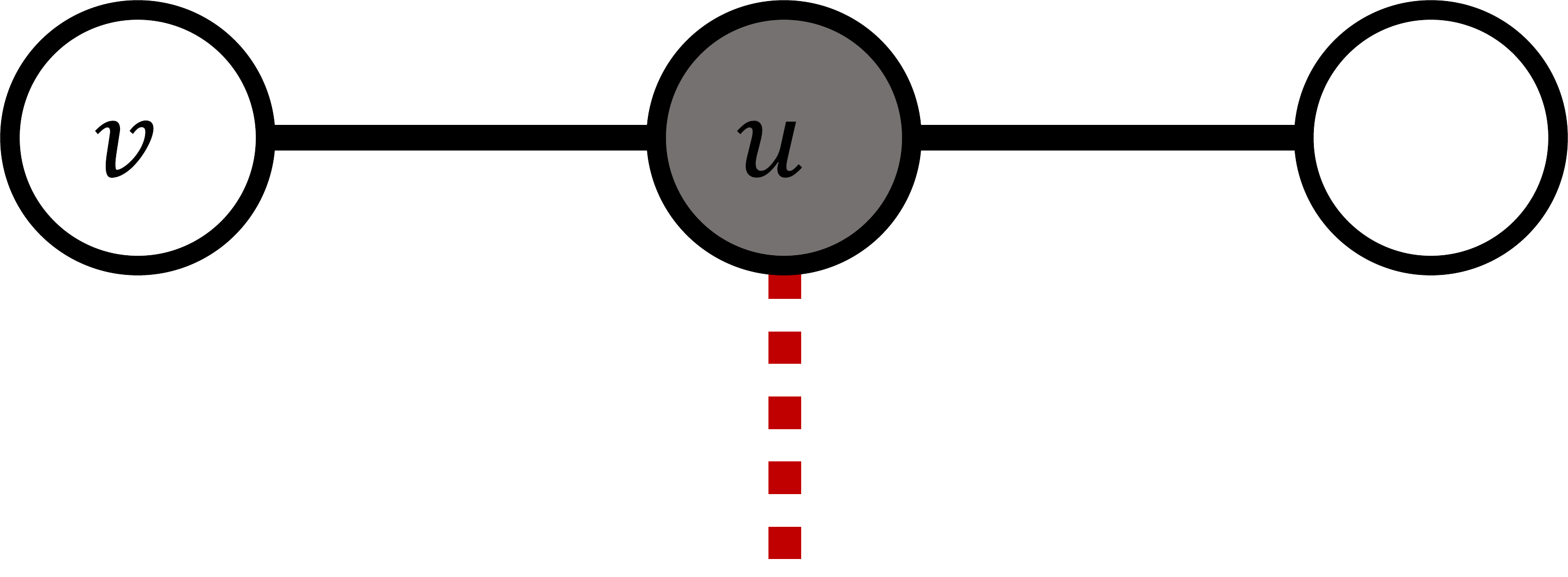}
    \caption{A vertex $v$ is affected because its neighbor $u$ is affected, and $v$ depends on $u$ contracting
    in order for the contraction to be maximal.}
    \label{fig:affect-by-dependence-example}
\end{figure}

\myparagraph{Affected vertices}
To perform the update efficiently, we define the notion of an \emph{affected vertex}.

\begin{definition}[Affected]
    A vertex is \emph{affected} at level $i$ if any of:
    \begin{enumerate}[leftmargin=*]
        \item it is alive in one of $F_i$ and $F_i'$ but not the other (which means it either contracted in an earlier round originally, but survived later after the update, or vice versa).
        \item it is alive in both $F_i$ and $F_i'$ but has a different adjacency list
        \item it is alive in both $F_i$ and $F_i'$, does not contract in $F_i$, but all of its neighbors $u$ in $F_i$ that contracted are affected.
    \end{enumerate}
\end{definition}

\noindent The first two cases of affected vertices are intuitive. If a vertex used to exist at round $i$ but no longer
does, or vice versa, it definitely needs to be updated in round $i$. If a vertex has a different adjacency
list than it used to, then it definitely needs to be processed because it can not possibly contract in the
same manner, or may change from being eligible to ineligible to contract or vice versa.

The third case of affection is more subtle, and is important for the correctness and efficiency of the
algorithm. Suppose an eligible vertex $v$ doesn't contract in round $i$. Since the contraction forms
a maximal independent set, at least one of $v$'s neighbors must contract. If all such neighbors are
affected, they may no longer contract after in the updated forest, which would leave $v$ uncontracted
and without a contracting neighbor, violating maximality. Therefore $v$ should also be considered in
the update. Figure~\ref{fig:affect-by-dependence-example} shows an example scenario where this is
important.

Note that by the definition, vertices only become affected because they have an affected neighbor either in the previous round or the same round. We call this \emph{spreading} affection.

\myparagraph{The algorithm} With the definition of affected vertices, the update algorithm can be summarized as stated in Algorithm~\ref{alg:update}.
Each level is processed sequentially, while the subroutines that run on each level are parallel. Line~\ref{line:new_affected} is implemented
by looking at each affected vertex of the previous round and any vertex within distance two of those in parallel, then filtering those which do not satisfy the definition of affected. Note that by the definition of affected, looking at vertices within distance two is sufficient since at worst, affection can only spread to neighbors and possibly those neighbors' uncontracted neighbors. Using Goldberg and Zwick's approximate compaction algorithm~\cite{goldberg1995optimal}, this step takes linear work in the number of affected vertices and $O(\log \log n)$ span.
In Section~\ref{sec:analysis}, we show that the number of affected vertices at each level is $O(k)$, so this is efficient.

Line~\ref{line:new_mis} is accomplished using Lemma~\ref{lem:mis-in-constant-degree-graph}. This finds a maximal independent set of eligible affected vertices in linear work and $O(\log^{(c)} n)$ span, so we can choose $c = 2$ and use the fact that there are $O(k)$ affected vertices to find the new maximal independent set in $O(k)$ work and $O(\log \log k)$ span.

Line~\ref{line:update_tree} is implemented by looking at each affected vertex in parallel and updating it to reflect its new behaviour in $F_i'$. This entails
writing the corresponding adjacent edges into the adjacency lists of its neighbors in round $i+1$ if it did not contract, or writing the appropriate cluster
values if it did. At the same time, for each contracted vertex, the algorithm computes the augmented value on the resulting RC cluster from the
values of the children.

\myparagraph{Correctness}
We now argue that the algorithm is correct.

\begin{lemma}
    After running the update algorithm, the contraction is still maximal, i.e., the contracted vertices at each level form
    a maximal independent set of degree one and two vertices.
\end{lemma}

\begin{proof}
    Call a vertex \emph{eligible} if it has degree one or two in $F_i'$ and is not adjacent to an unaffected vertex
    that contracts in $F_i$.
    Consider any eligible unaffected vertex $v \in F_i'$.  Since it is unaffected it exists in $F_i$. Suppose $v$
    contracts in $F_i$, then it still contracts in $F_i'$. We need to argue that no neighbor
    of $v$ contracts in $F_i'$. For any unaffected neighbor $u$, it didn't contract in $F_i$ and hence
    still doesn't contract in $F_i'$. If $u$ is an affected neighbor, it is not considered eligible and
    hence does not contract since $v$ is unaffected and contracted. Therefore none of $v$'s neighbors
    contract in $F_i'$.

    Now suppose $v$ doesn't contract in $F_i'$. Since it is unaffected it exists and doesn't contract
    in $F_i$. Therefore it has a neighbor $u \in F_i$ that contracts. If $u$ is unaffected, then $u \in F_i'$
    and contracts. If $u$ is affected, then $v$ is affected by dependence. Therefore all eligible unaffected vertices
    $v \in F_i'$ satisfy the invariant.

    Consider any eligible affected vertex $v \in F_i'$. If $v$ has an unaffected contracted neighbor, then
    $v$ does not contract in $F_i'$, and has a contracted neighbor. Otherwise, $v$ participates in the MIS, and has
    no prior contracted neighbor. Since the algorithm finds a maximal independent set on the candidates, $v$ either contracts and has
    no contracted neighbor, or doesn't contract and has a contracted candidate neighbor. Therefore $v$ satisfies the invariant.

    Together, we can conclude that every eligible vertex satisfies the invariant. Lastly, if a vertex is not eligible,
    then it satisfies the invariant since it either has degree greater than two and hence can not contract, or it is
    adjacent to a vertex that contracts.
\end{proof}

\noindent It remains to show that the algorithm is efficient, which we will do in Section~\ref{sec:analysis}.

\section{Analysis}\label{sec:analysis}

We start by proving some general and useful
lemmas about the contraction process, from which the efficiency of the static algorithm immediately follows,
and which will later be used in the analysis of the update algorithm.

\subsection{Round and tree-size bounds}\label{subsec:round-and-size-bounds}

\begin{lemma}\label{lem:maximal_contraction}
    Consider a tree $T$ and suppose a maximal indpendent set of degree one and two vertices is
    contracted via \emph{rake} and \emph{compress} contractions to obtain $T'$. Then
    \begin{equation}
        |T'| \leq \frac{5}{6}|T|
    \end{equation}
\end{lemma}

\begin{proof}
    More than half of the vertices in any tree have degree one or two~\cite{werneck2006design}. An MIS among them is at least one third
    of them since every adjacent run of three vertices must have at least one selected, so
    at least one sixth of the vertices contract. Therefore the new tree has at most $\tfrac{5}{6}$ as many vertices.
\end{proof}

\noindent The lemma also applies to forests since it can simply be applied independently to each component.
Three important corollaries follow that allow us to bound the cost of various parts of the algorithm.
Corollary~\ref{cor:completely-contract} gives the number of rounds required to fully contract a forest,
and Corollary~\ref{cor:contract-to-noverlogn} gives a bound on the number of rounds required to shrink
a forest to size $n / \log n$.  Lastly, Corollary~\ref{cor:contract-to-k} gives a bound on the number of rounds
required to shrink a tree to size $k$, which is useful in bounding the work of the batched update and query algorithms.

\begin{corollary}\label{cor:completely-contract}
    Given a forest on $n$ vertices, maximal tree contraction completely contracts a forest of $n$ vertices in $\log_{6/5} n$ rounds.
\end{corollary}


\begin{corollary}\label{cor:contract-to-noverlogn}
    Given a forest on $n$ vertices, after performing $\log_{6/5} \log n$
    rounds of maximal tree contraction, the number of vertices in the resulting forest is at most $n / \log n$.
\end{corollary}


\begin{corollary}\label{cor:contract-to-k}
    Given a forest on $n$ vertices and some integer $k \geq 1$, after performing $\log_{6/5}\left(1+n/k\right)$
    rounds of maximal tree contraction, the number of vertices in the resulting forest is at most $k$.
\end{corollary}


\begin{proof}
By Lemma~\ref{lem:maximal_contraction}, the number of vertices in each round is at most $5/6^\text{th}$s of the previous round, so the number
remaining after round $r$ is at most $n \left(5/6\right)^r$.  The three corollaries follow.
\end{proof}

\subsection{Analysis of the static algorithm}

Armed with the lemmas and corollaries of Section~\ref{subsec:round-and-size-bounds}, we can now
analyze the static algorithm.

\begin{theorem}
    The basic maximal tree contraction algorithm can be implemented in $O(n)$ work and $O(\log n \log \log n)$
    span for a forest of $n$ vertices.
\end{theorem}

\begin{proof}
    The work performed at each round is $O(|F_i|)$, i.e., the number of live vertices in the forest at that round.
    By Lemma~\ref{lem:maximal_contraction}, the total work is therefore at most
    \begin{equation}
        \sum_{i=0}^\infty n \left(\frac{5}{6}\right)^i = n \sum_{i=0}^\infty \left(\frac{5}{6}\right)^i = 6n.
    \end{equation}
    The span of the algorithm is $O(\log \log n)$ per round to perform the maximal independent set and approximate
    compaction operations by Lemmas~\ref{lem:mis-in-constant-degree-graph}~and~\ref{lem:approximate-compaction}.
    By Corollary~\ref{cor:completely-contract} there are $O(\log n)$ rounds, hence the total span is $O(\log n \log \log n)$.
\end{proof}

\noindent Section~\ref{sec:improved-span} explains how to improve the span of the algorithm to $O(\log n \log^{(c)} n)$ for any constant
$c$.

\subsection{Analysis of the update algorithm} The analysis of the update algorithm
follows a similar pattern to the analysis of the randomized change propagation algorithm~\cite{acar2020batch}.
We sketch a summarised version below, then present the full analysis.

\myparagraph{Summary}
We begin by establishing the criteria for vertices becoming affected. Initially, the endpoints of
the updated edges and a small neighborhood around them are affected. We call these the \emph{origin vertices}.
For each of these vertices, it may \emph{spread} its affection to nearby vertices in the next round.
Those vertices may subsequently spread to other nearby vertices in the following round and so on.
As affection spreads, the affected vertices form an \emph{affected component}, a connected set of
affected vertices whose affection originated from a common origin vertex. An affected vertex that is
adjacent to an unaffected vertex is called a \emph{frontier vertex}.  Frontier vertices are those which
are capable of spreading affection.  Note that it is possible that in a given round, a vertex that
becomes affected was adjacent to multiple frontier vertices of different affected components,
and is subsequently counted by both of them, and might therefore be double counted in the analysis.
This is okay since it only overestimates the number of affected vertices in the end.

With these definitions established, our results show that each affected component consists of at
most two frontier vertices, and that at most four new vertices can be added to each affected component
in each round. Given these facts, since a constant fraction of the vertices
in any forest must contract in each round, we show that the size of each affected component shrinks by a constant
fraction, while only growing by a small additive factor. This leads to the conclusion
that each affected component never grows beyond a \emph{constant size}, and since there are initially
$O(k)$ origin vertices, that there are never more than $O(k)$ affected vertices in any round.
This fact allows us to establish that the update algorithm is efficient.

\myparagraph{The proofs} We now prove the afformentioned facts.



\begin{lemma}\label{lem:unaffected_same_contraction}
    If $v$ is unaffected at round $i$, then $v$ contracts in round $i$ in $F$ if and
    only if $v$ contracts in round $i$ in $F'$.
\end{lemma}

\begin{proof}
    Unaffected vertices are ignored by the update algorithm, and hence remain the same
    before and after an update.
\end{proof}

\noindent If a vertex is not affected during round $i$ but is affected during round $i+1$, we say that
$v$ \emph{becomes affected in round $i$}. A vertex can become affected in two ways.

\begin{lemma}\label{lem:become_affected}
If $v$ becomes affected in round $i$, then one of the following is true:
\begin{enumerate}[leftmargin=*]
    \item $v$ has an affected neighbor $u$ at round $i$ which contracted in either $F_i$ or $F_i'$
    \item $v$ does not contract by round $i+1$, and has an affected neighbor $u$ at round $i+1$ that contracts in $F_{i+1}$.
\end{enumerate}
\end{lemma}

\begin{proof}
    First, since $v$ becomes affected in round $i$, it is not already affected at round $i$. Therefore,
    due to Lemma~\ref{lem:unaffected_same_contraction}, $v$ does not contract, otherwise it would do
    so in both $F$ and $F'$ and hence be unaffected at round $i+1$. Since $v$ does not contract, $v$
    has at least one neighbor, otherwise it would finalize.

    Suppose that (1) is not true, i.e., that $v$ has no affected neighbors that contract in $F_i$
    or $F_i'$ in round $i$. Then either none of $v$'s neighbors contract in $F_i$, or only unaffected
    neighbors of $v$ contract in $F_i$. By Lemma~\ref{lem:unaffected_same_contraction}, in either case,
    $v$ has the same set of neighbors in $F_{i+1}$ and $F_{i+1}'$. Therefore, since $v$ is affected
    in round $i+1$, does not contract in either $F_i$ or $F_i'$, and has the same neighbors in both,
    it must be in Case~(3) in the definition of affected. Therefore, $v$ has an affected neighbor that
    contracts in $F_{i+1}$.

    Since $\neg (1) \Rightarrow (2)$, we have that $(1) \lor (2)$ is true.
\end{proof}

\begin{definition}[Spreading affection]\label{def:spread}
    An affected vertex $u$ \emph{spreads to $v$} if $v$ was unaffected at the beginning of round $i$ and became
    affected in round $i$ (i.e., is affected in round $i+1$) because
    \begin{enumerate}[leftmargin=*]
        \item $v$ is a neighbor of $u$ at round $i$, and $u$ contracts in round $i$ in either $F_i$ or $F_i'$, or
        \item $v$ does not contract in round $i+1$, and is a neighbor of $u$, which contracts in round $i+1$.
    \end{enumerate}
\end{definition}

We call Case~(1), \emph{spreading directly} and Case~(2) \emph{spreading by dependence}. Figure~\ref{fig:directly-affect} shows two examples
of directly spreading affection. Figure~\ref{fig:affect-by-dependence} shows an example of spreading affection by dependence.
Our end goal is to bound the number of affected vertices at each level, since this corresponds to the amount of work
required to update the contraction after an edge update.

\begin{figure}
    \centering
    \includegraphics[width=0.55\columnwidth]{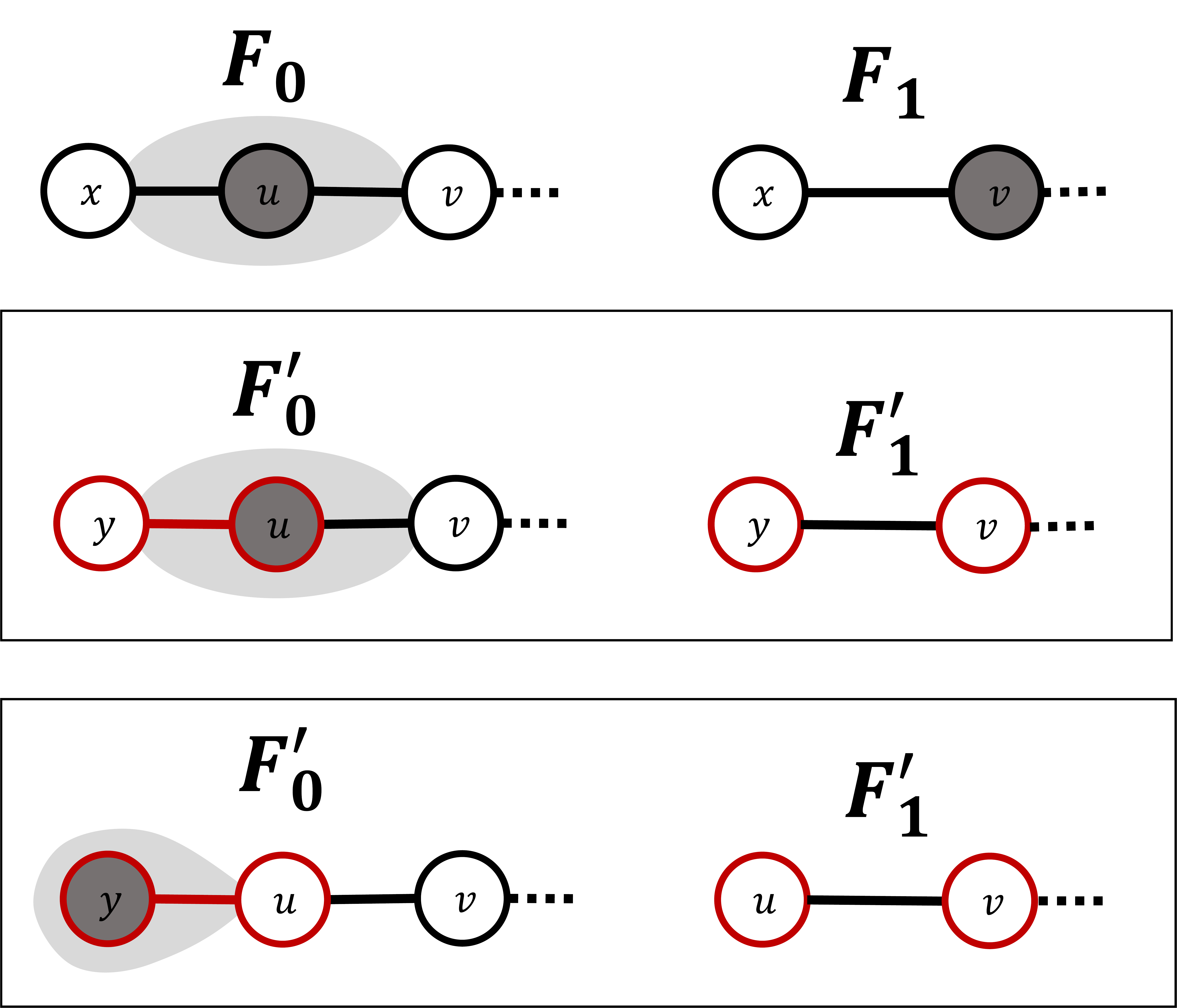}
    \caption{Direct affection (two possibilities): A vertex $u$ directly affects its neighbor $v$. $u$ is affected at round $0$ since its adjacency list
    was changed. Since $u$ contracts in $F_0$ and changes the adjacency list of $v$ at round $1$, $v$ becomes
    affected. Note that this can happen whether or not $u$ contracts in $F_0'$ as shown in the second possibility.}
    \label{fig:directly-affect}
\end{figure}

\begin{figure}
    \centering
    \includegraphics[width=0.65\columnwidth]{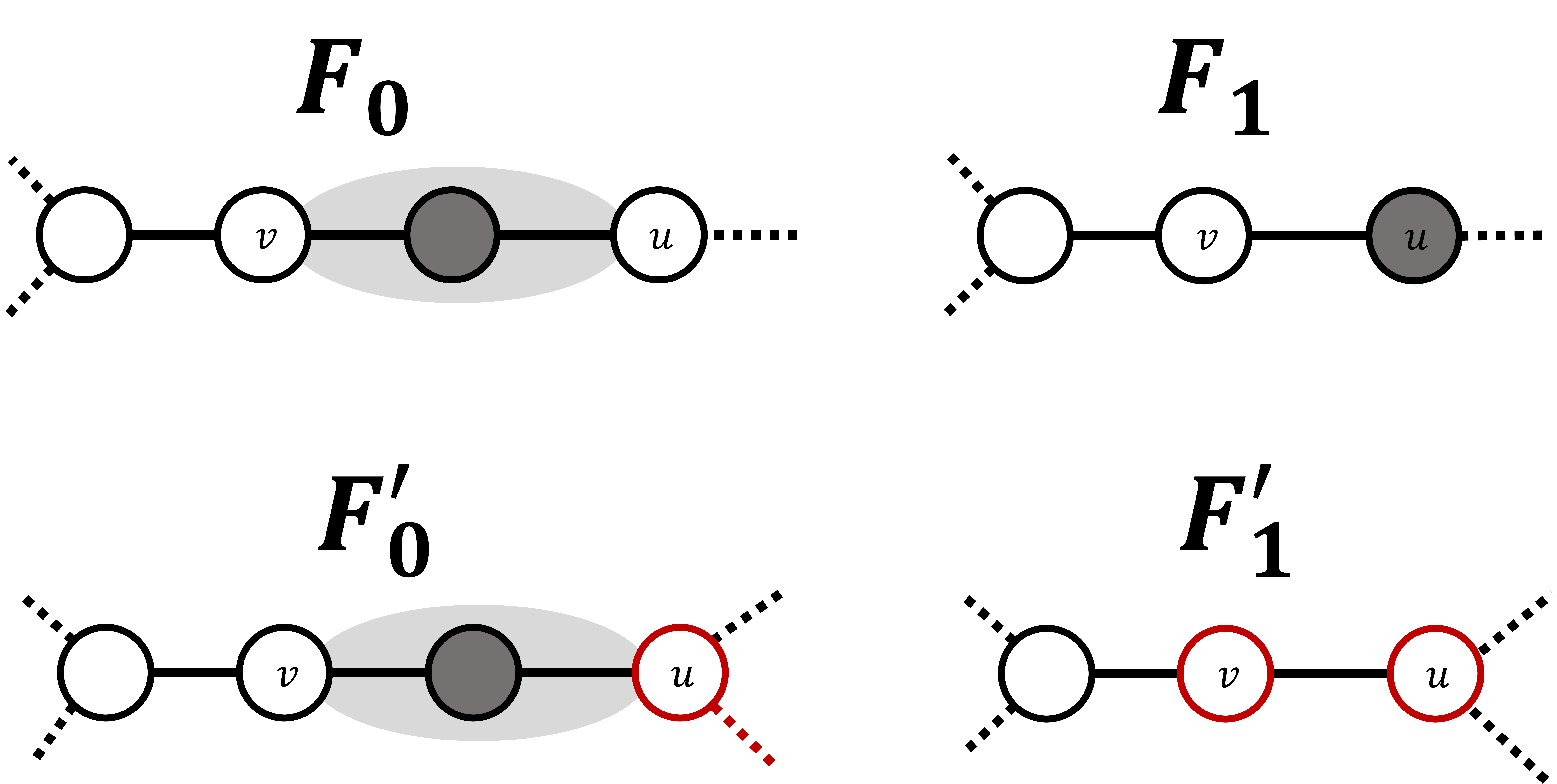}
    \caption{Affection by dependence: A vertex $u$ affects its neighbor $v$ by dependence. Even though $v$'s adjacency list
    is the same in both forests, it is affected because it depended on $u$ to contract in $F_1$ for maximality to be
    satisfied. Since $u$ can no longer contract in $F_1'$, it is important that $v$ is able to.}
    \label{fig:affect-by-dependence}
\end{figure}

Let $\mathcal{A}^i$ denote the set of affected vertices in round $i$.

\begin{lemma}\label{lem:initial_affected}
    For a batch update of size $k$ (insertion or deletion of $k$ edges), we have $|\mathcal{A}^0| \leq 6k$.
\end{lemma}

\begin{proof}
    A single edge changes the adjacency list of its two endpoints. These two endpoints might contract in
    the first round, which affects their uncontracted neighbors by dependence. However, vertices that contract
    have degree at most two, so this is at most two additional vertices per endpoint. Therefore there are up to
    $6$ affected vertices per edge modification, and hence up to $6k$ affected vertices in total.
\end{proof}


\noindent Each edge modified at round $0$ affects some set of vertices,
which spread to some set of vertices at round $1$, which spread to some set of vertices at round $2$ and so on.
We will therefore partition the set of affected
vertices into $s = |\mathcal{A}^0|$ \emph{affected components}, indicting the ``origin'' of the affection.  When a vertex
$u$ spreads to $v$, it will add $v$ to its component for the next round.

More formally, we will construct $\mathcal{A}_1^i, \mathcal{A}_2^i, \ldots, \mathcal{A}_s^i$, which form
a partition of $\mathcal{A}^i$. We start by arbitrarily partitioning $\mathcal{A}^0$ into $s$ singleton
sets $\mathcal{A}_1^0, \mathcal{A}_2^0, \ldots, \mathcal{A}_s^0$. Given $\mathcal{A}_1^i, \mathcal{A}_2^i, \ldots, \mathcal{A}_s^i$,
we construct $\mathcal{A}_1^{i+1}, \mathcal{A}_2^{i+1}, \ldots, \mathcal{A}_s^{i+1}$ such that $\mathcal{A}^{i+1}_j$ contains
the affected vertices $v \in A^{i+1}$ that were either already affected in $\mathcal{A}^i_j$ or were spread
to by a vertex $u \in \mathcal{A}^i_j$. Note that it is possible, under the given definition, for multiple vertices
to spread to another, so this may overcount by duplicating vertices.

\begin{definition}[Frontier]
    A vertex $v$ is a \emph{frontier} at round $i$ if $v$ is affected at round $i$ and one of its neighbors
    in $F_i$ is unaffected at round $i$.
\end{definition}

\begin{lemma}\label{lem:frontier-looks-same-in-both}
    If $v$ is a frontier vertex at round $i$, then it is alive in both $F_i$ and $F_i'$ at round $i$, and
    is adjacent to the same set of unaffected vertices in both.
\end{lemma}

\begin{proof}
    If $v$ were dead in both forests it would not be affected and hence not a frontier vertex. If $v$ were
    alive in one forest but dead in the other, then all of its neighbors
    would have a different set of neighbors in $F_i$ and $F_i'$ (they must be missing $v$) and hence all
    of them would be affected, so $v$ would have no unaffected neighbors and hence not be a frontier.

    Similarly, consider an unaffected neighbor $u$ of $v$ in either forest. If $u$ was not adjacent to $v$
    in the other forest, it would have a different set of neighbors and hence be affected.
\end{proof}

\noindent If a $v$ spreads to a vertex in round $i$, then clearly $v$ must be a frontier. Our next goal
is to analyze the structure of the affected sets and then show that the number of frontier vertices is small.

\begin{lemma}\label{lem:affected-tree}
    For all $i, j$, the subforest induced by $\mathcal{A}^i_j$ in $F_i$ is a tree
\end{lemma}

\begin{proof}
    The rake and compress operations both preserve the connectedness of the underlying tree, and Lemma~\ref{lem:become_affected}
    shows that affection only spreads to neighboring vertices.
\end{proof}

\begin{lemma}\label{lem:frontier_size}
    $\mathcal{A}^i_j$ has at most two frontiers and $|\mathcal{A}^{i+1}_j \setminus \mathcal{A}^i_j| \leq 4$.
\end{lemma}

\begin{proof}
We proceed by induction on $i$. At round $0$, each group contains one vertex, so it definitely
contains at most $2$ frontier vertices.
Consider some $\mathcal{A}^i_j$ and suppose it contains one frontier vertex $u$, which may spread
directly by contracting (Definition~\ref{def:spread}). If $u$ spreads directly, then it either compresses
or rakes in $F_i$ or $F_i'$. This means it has degree at most two in $F_i$ or $F_i'$, and by
Lemma~\ref{lem:frontier-looks-same-in-both}, it is therefore adjacent to at most two unaffected vertices,
and hence may spread to at most these two vertices. Since $u$ contracts, it is no longer a frontier by
Lemma~\ref{lem:frontier-looks-same-in-both}, but its newly affected neighbors may become frontiers, so
the number of frontiers is at most two.

Suppose $u$ spreads via dependency in round $i$ (Case 2 in Definition~\ref{def:spread})
in $\mathcal{A}^{i+1}_j$ and contracts in $F_{i+1}$. Since $u$ contracts in $F_{i+1}$, it has at most two neighbors, and by
Lemma~\ref{lem:frontier-looks-same-in-both}, it is also adjacent to at most two unaffected vertices,
and may spread to at most these two vertices. If it spreads to one of them, it may become a frontier
and hence there are at most two frontier vertices. If it spreads to both of them, $u$ is no longer
adjacent to any unaffected vertices and hence is no longer a frontier, so there are still at most
two frontier vertices, and $|\mathcal{A}^{i+1}_j \setminus \mathcal{A}^i_j| \leq 3$.

Now consider some $\mathcal{A}^i_j$ that contains two frontier vertices $u_1, u_2$. By Lemma~\ref{lem:affected-tree},
$u_1$ and $u_2$ each have at least one affected neighbor. If either contract, it would no longer be a frontier,
and would have at most one unaffected neighbor which might become affected and a frontier. Therefore the number
of frontiers is preserved when affection is spread directly.

Lastly, suppose $u_1$ or $u_2$ spreads via dependency in round $i$. Since it would contract in $F_{i+1}$, it has at most
one unaffected neighbor which might become affected and become a frontier. It would subsequently have no unaffected
neighbor and therefore no longer be a frontier. Therefore the number of frontiers remains at most two and
$|\mathcal{A}^{i+1}_j \setminus \mathcal{A}^i_j| \leq 4$.
\end{proof}

\noindent Now define $\mathcal{A}^i_{F,j} = \mathcal{A}^i_j \cap V^i_F$, the set of affected vertices from $\mathcal{A}^i_j$ that
are live in $F$ at round $i$, and similarly define $\mathcal{A}^i_{F',j}$ for $F'$.

\begin{lemma}\label{lem:affected-component-size}
    For every $i,j$ we have
    \begin{equation}
        |\mathcal{A}^i_{F,j}| \leq 26.
    \end{equation}
\end{lemma}

\begin{proof}
    Consider the subforest induced by the set of affected vertices $\mathcal{A}^i_{F,j}$. By Lemmas~\ref{lem:affected-tree}
    and~\ref{lem:frontier_size}, this is a tree with two frontier vertices. The update algorithm finds and contracts a
    maximal independent set of affected degree one and two vertices that are not adjacent to an unaffected vertex that
    contracts in $F_i$. There can be at most two vertices (the frontiers) that are adjacent to an unaffected vertex,
    and at most four new affected vertices appear by Lemma~\ref{lem:affected-tree},
    so by Lemma~\ref{lem:maximal_contraction}, the size of the new affected set is
    \begin{equation}
    \begin{split}
        |\mathcal{A}^{i+1}_{F,j}| &\leq 4 + \frac{5}{6}\left( |\mathcal{A}^i_{F,j}| - 2 \right) + 2 \\
        &= \frac{26}{6} + \frac{5}{6} |\mathcal{A}^i_{F,j}|
    \end{split}
    \end{equation}
    Since $|\mathcal{A}^0_{F,j}| = 1$, we obtain
    \begin{equation}
        |\mathcal{A}^{i+1}_{F,j}| \leq \frac{26}{6} \sum_{r=0}^\infty \left(\frac{5}{6}\right)^r = \frac{\frac{26}{6}}{1 - \frac{5}{6}} = \frac{\frac{26}{6}}{\frac{1}{6}} = 26
    \end{equation}
\end{proof}

\begin{lemma}\label{lem:affected_size}
    Given a batch update of $k$ edges, for every $i$
    \begin{equation}
        |\mathcal{A}^i| \leq 312 k
    \end{equation}
\end{lemma}

\begin{proof}
    By Lemma~\ref{lem:initial_affected}, there are at most $6k$ affected groups. At any level, every affected
    vertex must be live in either $F$ or $F'$, so $\mathcal{A}^i_j = \mathcal{A}^i_{F,j} \cup \mathcal{A}^i_{F',j}$, and hence
    \begin{equation}
        |\mathcal{A}^i| \leq \sum_{j=1}^{6k} \left( |\mathcal{A}^i_{F,j}| + |\mathcal{A}^i_{F',j}| \right) \leq 6k \times 26 \times 2 = 312k
    \end{equation}
\end{proof}

\noindent We conclude that given an update of $k$ edges, the number of affected vertices at each level is $O(k)$.

\myparagraph{Putting it all together}
Given the series of lemmas above, we now have the power to analyze the performance of the update algorithm.

\begin{theorem}[Update performance]
    A batch update consisting of $k$ edge insertions or deletions takes $O\left(k \log\left(1 + n/k\right)\right)$ work
    and $O(\log n \log \log k)$ span.
\end{theorem}

\begin{proof}
    The update algorithm performs work proportional to the number of affected vertices at each level.  Consider separately the work performed processing the levels up to and including round $r = \log_{6/5}\left(1+n/k\right)$. By Lemma~\ref{lem:affected_size}, there
    are $O(k)$ affected vertices per level, so the work performed up to including level $r$ is
    \begin{equation}
        O(kr) = O\left(k \log\left(1 + \tfrac{n}{k}\right)\right).
    \end{equation}
    By Corollary~\ref{cor:contract-to-k}, after $r$ rounds of contraction, there are at most $k$ live vertices remaining in $F_r$ or $F_r'$.
    The number of affected vertices is at most the number of live vertices in either forest, and hence at most $2k$. The amount of affected vertices
    in all subsequent rounds is therefore at most
    \begin{equation}
    \begin{split}
        \sum_{i = 0}^\infty \left(\frac{5}{6}\right)^i 2k = \frac{2k}{1-\frac{5}{6}} = 12k,
    \end{split}
    \end{equation}
    and hence the remaining work is $O(k)$. Therefore the total work across all rounds is at most
    \begin{equation}
        O\left(k \log\left(1 + \tfrac{n}{k}\right)\right) + O(k) = O\left(k \log\left(1 + \tfrac{n}{k}\right)\right).
    \end{equation}
    In each round, it takes $O(\log^{(c)} k)$ span to find a maximal independent set of the affected vertices for
    any constant $c$, so we can choose $c = 2$ to match the span of approximate compaction required to filter out
    the vertices that are no longer affected in the next round. Each round takes $O(\log \log k)$ span, so over $O(\log n)$ rounds,
    this results in $O(\log n \log k \log k)$ span.
    
\end{proof}

\section{Optimizations}\label{sec:improved-span}

Our static tree contraction algorithm and our basic result on dynamically updating
it are work-efficient ($O(n)$ and $O(k \log(1+n/k))$ work respectively) and run in
$O(\log n \log \log n)$ and $O(\log n \log \log k)$ span respectively. In both cases,
there are two bottlenecks to the span: computing a maximal independent set, and
performing approximate compaction to remove vertices that have contracted or are
no longer affected. Improving the span of the maximal indpendent set is easy since
it runs in $O(\log^{(c)} n)$ span for any $c$, and we can just choose a smaller $c$
(the basic algorithm chose $c=2$ to match the span of approximate compaction).

Therefore, the only remaining bottleneck is the approximate compaction, which we
can improve as follows. We first describe a faster static algorithm, which introduces
the techniques we will use to improve the update algorithm. We also describe an improvement
that eliminates the need for concurrent writes since approximate compaction requires
the Common CRCW model.

Lastly, we will also show that our span optimization technique can be used to speed
up the randomized variant of the algorithm. 

\subsection{A lower span static algorithm}

The basic static algorithm uses approximate compaction after each round to filter out
the vertices that have contracted. This is important, since without this step,
every round would take $\Theta(n)$ work, for a total of $\Theta(n \log n)$ work,
which is not work efficient. This leads to an $O(n)$ work and $O(\log n \log \log n)$
span algorithm in the Common CRCW model using the $O(\log \log n)$-span approximate
compaction algorithm of Goldberg and Zwick~\cite{goldberg1995optimal}.
We can improve the span easily as follows by splitting the algorithm into two
phases.

\myparagraph{Phase One} Note that the purpose of compaction is to avoid performing wasteful work on
dead vertices each round. However, if the forest being contracted has just $O(n / \log n)$ vertices, then a ``wasteful'' algorithm
which avoids performing compaction takes at most $O(n)$ work anyway. So, the strategy for phase one is to contract the forest to size
$O(n / \log n)$, which, by Corollary~\ref{cor:contract-to-noverlogn} takes at most
$O(\log \log n)$ rounds. This is essentially the same strategy used by Gazit, Miller, and Teng~\cite{gazit1988optimal}. The work of the first phase is therefore $O(n)$ and the span,
using approximate compaction, is
\begin{equation}
    O\left( \left( \log \log n \right)^2 \right) + O\left( \log \log n \log^{(c)} n \right) = O\left( (\log \log n)^2 \right).
\end{equation}

\myparagraph{Phase Two}  In the second phase, we run the ``wasteful'' algorithm, which is simply the same algorithm
but not performing any compaction. Since the forest begins with $O(n / \log n)$ vertices
in this phase, this takes $O(n)$ work and completes in $O(\log n)$ rounds. Since
the span bottleneck is finding the maximal indepedent set in each round, the span is $O(\log n \log^{(c)} n)$
for any constant $c$.

Putting these together, the total work is $O(n)$, and the span is
\begin{equation}
  O\left( (\log \log n)^2 \right) + O\left( \log n \log^{(c)} n \right) = O\left( \log n \log^{(c)} n \right).
\end{equation}

\subsection{Eliminating concurrent writes}
The above optimized algorithm still uses approximate compaction which requires
the power of the Common CRCW model. We now breifly describe a variant without
this requirement. Phase Two is the same since it performs no compaction, so
we just have to improve Phase One. We do so by partitioning the vertices
into $n / \log n$ groups of size $O(\log n)$ by their identifier. Each round,
the algorithm simply considers each group and each vertex within each group in
parallel. After performing each round of contraction, each group independently
filters the vertices that contracted. We do so using an exact filter algorithm
instead of approximate compaction, but since each group has size $O(\log n)$,
the span is still $O(\log \log n)$ without requiring concurrent writes.

Since there are $n / \log n$ groups, each round takes an additional $O(n / \log n)$
work, but over $\log \log n$ rounds, this amounts to less than $O(n)$ additional
work, so the algorithm is still work efficient. At the end of the phase, collect
the vertices back into a single group in $O(n)$ work and $O(\log n)$ span,
then proceed with Phase Two.

\subsection{A lower span dynamic algorithm}

The span of the dynamic algorithm is also bottlenecked by the span of approximate
compaction, which is used on the affected vertices each round to remove vertices
that are no longer affected. We optimize the dynamic algorithm similarly to the
static algorithm, by splitting it into three phases this time.

\myparagraph{Phase One}  The algorithm will run Phase One for $\log_{6/5}\left(1 + n/k \right)$ rounds.
Note importantly that this depends on the batch size $k$, so the number of rounds
each phase runs is not always the same for each update operation.

Similarly to the optimized static algorithm without concurrent writes, we attack
the problem by splitting the affected vertices into groups. Specifically, we will
group the affected vertices into \emph{affected components} based on their origin
vertex as defined in Section~\ref{sec:analysis}.  There are $O(k)$ affected components,
each of which is initially a singleton defined by an affected vertex at round $0$.

In each round, the algorithm processes each affected component and each affected
vertex within in parallel. At the end of the round, the newly affected vertices
for the next round are identified for each component. Note that there could be
duplicates here since it is possible for multiple neighbors of a vertex to
spread to it at the same time. To tiebreak, and ensure that only one copy of an
affected vertex exists, if multiple vertices spread to the same vertex, only the one
with the lowest identifier adds the newly affected vertex to its component. Since the forest has
constant degree, this can be checked in constant time.

Given the set of affected vertices, new and old, we can then filter each component
independently in parallel to remove vertices that are no longer affected in the next round.
The critical insight is that according to Lemma~\ref{lem:affected-component-size}, each affected
component has \emph{constant size}, so this filtering takes constant work and span!

Having to maintain this set of $k$ affected components adds an additional $O(k)$
work to each round, but since we run Phase One for only $O(\log(1 + n/k))$ rounds,
this is still work efficient.

\myparagraph{Phase Two}
According to Corollary~\ref{cor:contract-to-k}, by the time Phase Two begins, the
forest will have contracted to the point that at most $k$ vertices remain. From this
point onwards, we use an algorithm very similar to the static algorithm to complete
the remaining rounds, and thus split into two more phases. First, we can collect
the contents of each of the $O(k)$ affected components back into a single array
of $O(k)$ affected vertices. This can be done in at most $O(k)$ work and $O(\log k)$ span.

Given an array of $O(k)$ affected vertices, we logically partition it into $k / \log k$
groups of size $O(\log k)$. We then run the basic dynamic update algorithm for $\log \log k$
rounds, using a filter algorithm (not approximate compaction) at each round to remove
vertices that are no longer affected. The span of this phase is therefore $O((\log \log k)^2)$,
and costs at most $O(k)$ additional work.

\myparagraph{Phase Three}
After completing Phase Two, by Corollary~\ref{cor:contract-to-noverlogn}, there can be at
most $O(k / \log k)$ vertices alive in the forest, and hence at most twice that many
affected vertices (affected vertices may be alive in either the new or old forest).
Phase Three simply collects the remaining affected vertices and performs the same
steps as Phase One. We create up to $O(k / \log k)$ singleton affected components,
and then in each round, process each vertex in each component in parallel, then spread
to any newly affected vertices. Each affected component remains constant size by Lemma~\ref{lem:affected-component-size}
and the work performed in each round is at most $O(k / \log k)$ for $O(\log k)$ rounds,
a total of $O(k)$ work.  Since each affected component is constant size, maintaining them
takes constant time. After $O(\log k)$ rounds, the forest is fully contracted.

In total, at most $O(k \log(1 + n/k))$ additional work is added, so the algorithm is still
work efficient. The span of Phase One and Three is dominated by the span of computing the maximal
independent set, so the final resulting span is now
\begin{equation}
\begin{split}
    & O\left( \log\left(1+\frac{n}{k}\right) \log^{(c)} k + \left( \log \log k \right)^2 + \log(k) \log^{(c)} k\right),  \\
    =&\ O\left( \log n \log^{(c)} k \right).
\end{split}
\end{equation}

\subsection{A lower span randomized algorithm}

With the span optimization above, the bottleneck of what remains is entirely due to the subroutine for finding
the maximal independent set, which takes $O(\log^{(c)} k)$.  Our optimization essentially removes the span causes
by compaction. In the randomized variant of the algorithm~\cite{acar2020batch}, the span is $O(\log n \log^* n)$,
where the $\log^* n$ factor also comes from performing approximate compaction (which is $O(\log^* n)$ when randomization is allowed).
In the randomized variant, however, finding the independent set takes
constant span rather rather than $O(\log^{(c)} k)$.  It works by raking all the leaves, then flipping a coin for each
vertex and compressing the vertices that flip heads while its two neighbors flip tails, which happens with $1/8$ probability.

It can therefore be shown that a constant fraction of the vertices contract on each round, and that the contraction process
takes $O(\log n)$ rounds with high probability. We can therefore substitute our deterministic maximal independent set with the randomized
variant and use Acar et al.'s~\cite{acar2020batch} definition of affected vertices to obtain a more efficient randomized algorithm.
Their analysis implies that the resulting algorithm is work efficient, running in $O(k \log(1+n/k))$ expected work for a batch
of $k$ updates, and in $O(\log n)$ span.

\section{Conclusion}

We presented the first deterministic work-efficient parallel algorithm for the batch-dynamic trees problem. We showed
that parallel Rake-Compress Trees~\cite{acar2005experimental,acar2020batch} can be derandomized using a variant of parallel tree
contraction that contracts a deterministic maximal independent set of degree one and two vertices.
Our algorithm performs $O(k \log(1 + n/k))$ work for a batch of $k$
updates and runs in $O(\log n \log^{(c)} k)$ span for any constant $c$. We also applied our techniques to improve the span of the
randomized variant from $O(\log n \log^* n)$ to just $O(\log n)$, and showed that other batch-dynamic graph problems
can be solved deterministically.

Several interesting questions still remain open. Our deterministic algorithm requires $O(\log n \log^{(c)} k)$ span, while our
improvement of the randomized variant requires just $O(\log n)$. Can we obtain a deterministic algorithm with $O(\log n)$ span?
It seems unlikely that the exact algorithm that we present here could be optimized to that point, since that would imply finding
a maximal independent set in $O(1)$ span work efficiently, and the fastest known algorithms run in $O(\log^* n)$ span but are not
even work efficient.  This doesn't rule out using other techniques instead of a maximal independent set, however.  The tree
contraction needs only to have the property that it contracts a constant fraction of the vertices in any subtree in order to
obtain our bounds, so any constant ruling set would suffice if one could compute it in $O(1)$ span.

Prior algorithms for deterministic parallel tree contraction are based on Cole and Vishkin's deterministic coin tossing technique~\cite{cole1986deterministic}
(which happens to be a subroutine used by our maximal independent set algorithm). It would be interesting to investigate whether
this could be used to obtain a more efficient dynamic algorithm.
Lastly, can our deterministic RC-Trees be used to derandimize other existing algorithms?

\bibliographystyle{abbrv}
\bibliography{ref}

\clearpage

\appendix

\section{Appendix: Derandomizing downstream results}

\subsection{Batch-dynamic graph connectivity}

Acar, Anderson, Blelloch, and Dhulipala~\cite{acar2019parallel} give a parallel batch-dynamic algorithm for general
graph connectivity which runs in $O(k \log n \log(1 + n/\Delta))$ amortized expected work and $O(\polylog n)$ span for an
average batch size of $\Delta$. Inside their algorithm, they use a parallel batch-dynamic Euler Tour
tree to maintain a set of spanning forests. If the input graph is ternarized so that a Rake-Compress Tree can be
used instead, we can derandomize the algorithm by substituting the Euler Tour tree for one.

In addition to Rake-Compress trees, the algorithm uses just a few other pieces of randomization. When a batch of
edges is inserted, it uses a semisort to group the edges by endpoint. We can instead use a regular sort which
takes $O(k \log n)$ work to perform this step since that doesn't increase the total work. It also computes a
spanning forest of the newly added edges with the existing connected components. This takes $O(k \alpha(n,m))$
work using the most efficient deterministic algorithm~\cite{cole1991approximate}, which still fits the work bounds.

The deletion algorithm is more complex, and involves performing a spanning forest computation on a set of
replacement candidate edges over a series of layers. Since spanning forests cannot be computed work efficiently
deterministically, this will incur an overhead of $O(\alpha(n,m))$~\cite{cole1991approximate}. Therefore the algorithm runs in
\begin{equation}
    O\left( k \log n \log\left( 1 + \frac{n}{\Delta} \right) \alpha(n,m) \right),
\end{equation}
amortized work and $O(\polylog n)$ span, i.e., just a factor of $\alpha(n,m)$ more work than the randomized algorithm.
The discovery of a work-efficient deterministic spanning forest algorithm (though this has been open for 20 years)
would make it match the work bounds.

\subsection{Batch-incremental MST}

Anderson, Blelloch, and Tangwongsan present an algorithm for batch-incremental minimum spanning trees which can insert
a batch of $k$ edges in $O(k \log(1 + n/k))$ work and $O(\polylog n)$ span. Underneath, their algorithm maintains the
MST using a (randomized) RC-Tree. It also however uses Cole, Klein, and Tarjan's linear-work parallel MST algorithm as a subroutine,
which there is no known work-efficient deterministic equivalent for. Derandomizing it will therefore also incur a penalty of $O(k \log \log n)$
work to pay for the MST using the algorithm of Chong et al.~\cite{chong2003improving}.

The final work of the algorithm after applying this penalty is therefore
\begin{equation}
    O\left( k \log\left( 1 + \frac{n}{k} \right) + k \log \log n \right).
\end{equation}

\end{document}